\documentclass[10pt,journal,compsoc]{IEEEtran}
\usepackage{amsmath,amsfonts,amsthm}
\usepackage{algorithmic}
\usepackage{array}
\usepackage[caption=false,font=normalsize,labelfont=sf,textfont=sf]{subfig}
\usepackage{textcomp}
\usepackage{stfloats}
\usepackage{epstopdf}
\usepackage[shortcuts]{extdash}
\usepackage[figuresright]{rotating}
\usepackage{url}
\usepackage{dsfont}
\usepackage{diagbox}
\usepackage{multirow}
\usepackage{verbatim}
\usepackage{graphicx}
\def\BibTeX{{\rm B\kern-.05em{\sc i\kern-.025em b}\kern-.08em
    T\kern-.1667em\lower.7ex\hbox{E}\kern-.125emX}}
\usepackage{balance}
\newtheorem{theorem}{Theorem}[section]

\newtheorem{lemma}{Lemma}[section]

\ifCLASSOPTIONcompsoc
  \usepackage[nocompress]{cite}
\else
  \usepackage{cite}

\fi
\ifCLASSINFOpdf
\else
\fi
\begin{document}
%
\title{Quantum Support Vector Machine without Iteration}
%
%
%
%

\author{Rui~Zhang,
        Jian~Wang*,
        Nan~Jiang**,
        and~Zichen~Wang,
\IEEEcompsocitemizethanks{\IEEEcompsocthanksitem Rui Zhang and Jian Wang are with the Beijing Key Laboratory of Security and Privacy in Intelligent Transportation, Beijing Jiaotong University, Beijing 100044, China.
\protect\\
E-mail: wangjian@bjtu.edu.cn
\IEEEcompsocthanksitem Nan Jiang and Zichen Wang are with the Faculty of Information Technology, Beijing University of Technology, Beijing 100124, China. \protect\\
E-mail: e-mail: jiangnan@bjut.edu.cn.}
\thanks{Manuscript received April , ; revised August , .}}

%
%

\markboth{Journal of \LaTeX\ Class Files,~Vol.~, No.~, ~}%
{Shell \MakeLowercase{\textit{et al.}}: Quantum Support Vector Machine without Iteration}
%



\IEEEtitleabstractindextext{%
\begin{abstract}
Quantum algorithms can enhance machine learning in different aspects. In 2014, Rebentrost $et~al.$ constructed a least squares quantum support vector machine (LS-QSVM), in which the Swap Test plays a crucial role in realizing the classification. However, as the output states of a previous test cannot be reused for a new test in the Swap Test, the quantum algorithm
LS-QSVM has to be repeated in preparing qubits, manipulating operations, and carrying out the measurement.
This paper proposes a QSVM based on the generalized quantum amplitude estimation (AE-QSVM) which gets rid of the constraint of repetitive processes and saves the quantum resources.
At first, AE-QSVM is trained by using the quantum singular value decomposition. Then, a query sample is classified by using the generalized quantum amplitude estimation in which high accuracy can be achieved by adding auxiliary qubits instead of repeating the algorithm.
The complexity of AE-QSVM is reduced to
$O(\kappa^{3}\varepsilon^{-3}(log(mn)+1))$ with an accuracy $\varepsilon$,
where $m$ is the number of training vectors, $n$ is the dimension of the feature space, and $\kappa$ is the condition number. Experiments demonstrate that AE-QSVM is advantageous in terms of training matrix,
the number of iterations, space complexity, and time complexity.
\end{abstract}

\begin{IEEEkeywords}
Quantum support vector machine, quantum amplitude estimation, quantum singular value decomposition, quantum inner estimation, quantum computing.
\end{IEEEkeywords}}

\maketitle

\IEEEdisplaynontitleabstractindextext

%
\IEEEpeerreviewmaketitle

\IEEEraisesectionheading{\section{Introduction}\label{sec:introduction}}

%
%
%
%
\IEEEPARstart{I}{n} recent years, machine learning techniques have become powerful tools for finding patterns in data, for instance image recognition \cite{Yeung2021Image}, automated driven cars \cite{Zhang2021DDE}, network security \cite{He2021Research}, and \emph{etc}. 
The development of machine learning has rapidly increased the demand for computing power in hardware.
As the spacing of transistors approaches the physical limit of process manufacturing, it has been a bottleneck that researchers intend to improve the operating capability of the classical computer.
More powerful ways of processing information are needed as the amount of data generated in our society is growing.
Quantum computation is a promising new paradigm for performing fast computations, with the experimental demonstrations of quantum supremacy marked as the latest milestone \cite{Nielsen2000Quantum, Arute2019Supplementary, Zhong2021Quantum, Xin2020IEEEVPQC}.

The study of quantum computation originated in the 1980s. In 1982, Feynman \cite{Feynman1982Simulating} was the first to proposed the concept of quantum computation. After that the integer factoring problem \cite{Shor1994Algorithms} and the database search algorithm \cite{Grover1996fast} were essential evidences supporting the power of quantum computation. During the first decade of the 21st century, quantum computation appeared in various areas of computer science such as cryptography \cite{Gisin2001Quantum,Salim2020Enhancing,Kulik2022Experimental}, signal processing \cite{Zhou2015Quantum, Jiang2016Quantum, Zhang2022Boundary,Li2018Quantum}, and information theory \cite{Bennett2008Quantum}.

Quantum machine learning (QML) is an emerging research area
that attempts to harness the power of quantum information
processing to obtain speedups for classical machine learning
tasks.
Despite the fact that quantum machine learning is a recently surging field, it already encompasses a rich set of quantum techniques and approaches, for example
linear regression \cite{Schuld2016Prediction,Wang2017Quantum,Yu2021An}, clustering analysis \cite{Esma2013Quantum}, dimensionality reduction \cite{Romero2017Quantum,Yu2019Quantum}, data classification \cite{Cong2016Quantum,Schuld2017Implementing,2018Quantum,Dang2018Image}, and neural networks \cite{Li2020Quantumneural,Zhao2019Building,Joshi2021Entanglement}.
Besides, quantum machine learning algorithms have been applied to channel discrimination \cite{Xin2020IEEEVPQC}, vehicle classification \cite{Yu2008Quantum}, and image classification \cite{Cavallaro2020Approaching}.

Support vector machine (SVM) is a supervised machine learning technique for solving classification.
In recent years, there have been many studies about quantum SVM (QSVM) \cite{Allcock2020A, Kerenidis2021Quantum, Li2015Experimental, Havenstein2018Comparisons, Ye2020Quantum, Willsch2020Support}, it provides quadratic speedup and was first proposed by Anguita $et~al.$ \cite{Anguita2003Quantum}.
Moreover, the least squares QSVM (LS-QSVM) \cite{Rebentrost2014Quantum} given by Rebentrost $et~al.$ provide an exponential speedup compared with the classical algorithm.

In LS-QSVM and its descendants,
the classification was realized by using the Swap Test \cite{Buhrman2001Quantum,Garcia-Escartin2013swap}.
The output state is obtained by measuring the ancillary qubit in the Swap Test.
See Fig. \ref{fig1}, $|0\rangle$, $|\rho_{1}\rangle$, and $|\rho_{2}\rangle$ are the input states. For the outcome $|0\rangle$, we
have an entangled state $\frac{|\rho_{1}\rangle|\rho_{2}\rangle+|\rho_{2}\rangle|\rho_{1}\rangle}{\sqrt{2}}$ and for the outcome $|1\rangle$,
$\frac{|\rho_{1}\rangle|\rho_{2}\rangle-|\rho_{2}\rangle|\rho_{1}\rangle}{\sqrt{2}}$. In both cases, it is impossible to completely
separate the input states for a second time.
That is to say, Swap Test is destructive \cite{Garcia-Escartin2013swap}.
In LS-QSVM, however, the success probability $P$ can be obtained to an accuracy $\varepsilon$ by iterating $O(P(1-P)/\varepsilon^{2})$ times of the Swap Test \cite{Rebentrost2014Quantum}.
Hence, if the accuracy $\varepsilon$ is achieved, the completed algorithm must be carried out repeatedly; thus, resulting in a high consumption in qubit and time.
\begin{figure}
  \centering
  \includegraphics[width=8.5cm]{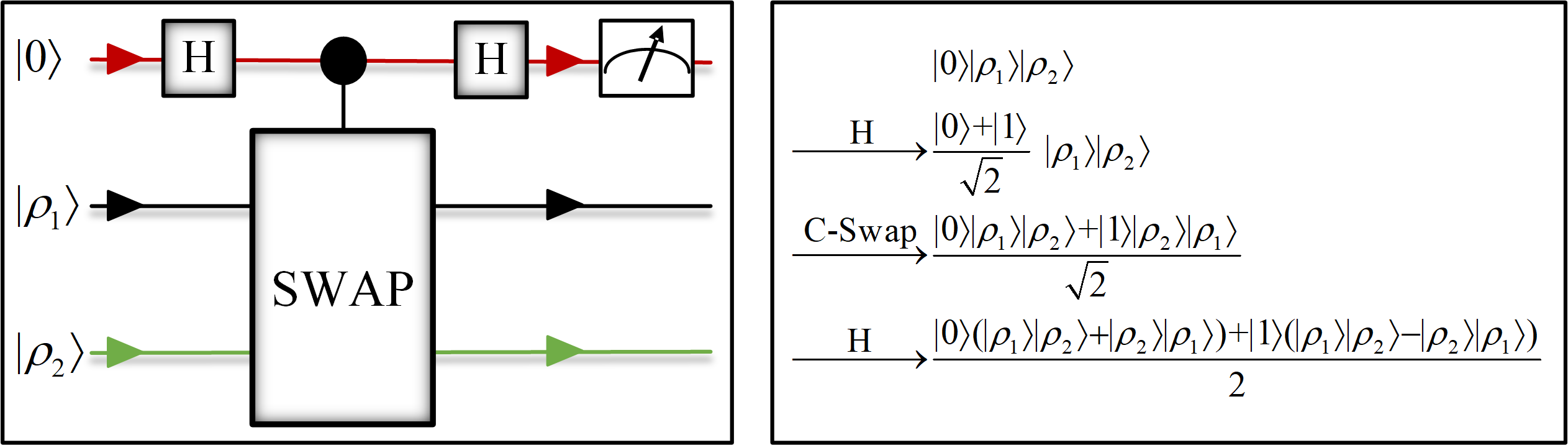} \\
  \caption{The quantum circuit of the Swap Test and its evolution}
  \label{fig1}
\end{figure}

To be specific, as shown in Fig. \ref{fig2}, LS-QSVM inherits the pattern of the classical least squares SVM (LS-SVM), i.e., the sample set was initially trained to obtain the parameter in LS-QSVM and the new sample was classified by using this parameter.
In cases where classical sampling algorithms require polynomial time, an exponential speedup is obtained in LS-QSVM \cite{Rebentrost2014Quantum, Bishwas2018An, Bishwas2016Big, Windridge2018Quantum, Feng2019Quantum}.
However, the speedup was achieved only for one time in training and classification.
As analyzed above, the quantum state collapses due to the use of the Swap Test. Therefore, we have to repeatedly train the sample and carry out a classification to get the classification label with a high accuracy.
\begin{figure}
  \centering
  \includegraphics[width=8.5cm]{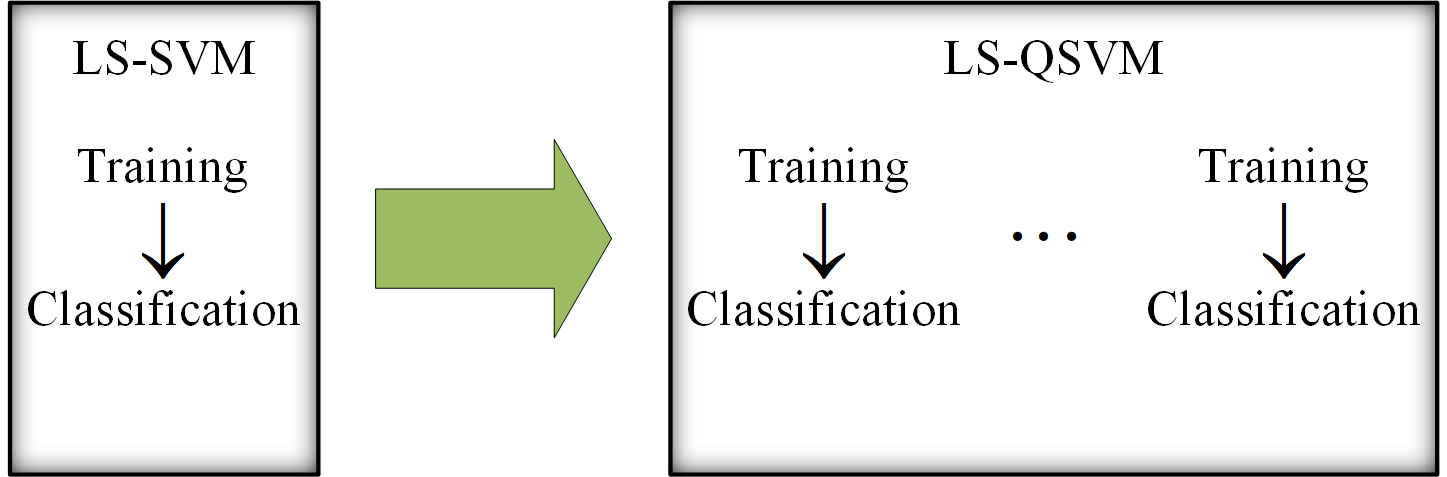} \\
  \caption{The comparation of LS-SVM and LS-QSVM}
  \label{fig2}
\end{figure}

In this work, we make progress for the challenge described above and propose a quantum support vector machine based on amplitude
estimation (AE-QSVM).
At first, the training process of AE-QSVM is realized
by using the quantum singular value decomposition.
Then, we present a method to calculate the inner product based on quantum amplitude estimation to classify the new sample.
Compared with LS-QSVM whose complexity is $O(\frac{\kappa^{3}\varepsilon^{-3}log(mn)+log(n)}{12\varepsilon^{2}})
$, the complexity of AE-QSVM is reduced to
$O(\kappa^{3}\varepsilon^{-3}(log(mn)+1))$ with an accuracy $\varepsilon$,
where $m$ is the number of training vectors, $n$ is the dimension of the feature space, and $\kappa$ is the condition number.

The remainder of the paper is organized as follows:
The related works are given in Section \ref{S2}. We present a general algorithm for quantum amplitude
estimation in Section \ref{S3}. Section \ref{S4} presents AE-QSVM. Simulations are presented in Section \ref{S5}. In Section \ref{S6}, we summarize this paper and discuss further research for AE-QSVM. Appendix reviews the basic concepts of quantum computation.

\section{Related works}\label{S2}
\subsection{Quantum support vector machine}
As shown in \cite{Rebentrost2014Quantum}, the QSVM is based on a least squares version of the
classical SVM \cite{Suykens1999LeastSquares}.

The task for the SVM is to
classify a vector into one of two classes, given $m$ training
data points of the form $\{(\boldsymbol x_{k}, y_{k}): \boldsymbol x_{k}\in \mathds{R}^{n}, y_{k}=\pm1\}_{k=1,\ldots,m}$,
where $y_{k}=1$ or $-1$, depending on the class to which $\boldsymbol x_{k}$
belongs.
In LS-SVM, the problem after applying the optimization of the Lagrangian can be formulated as a linear equation:
\begin{equation}\label{1}
F\begin{pmatrix}b\\ \boldsymbol \alpha\end{pmatrix}=
\begin{pmatrix}0&\boldsymbol 1^{\mathrm{ T }}\\ \boldsymbol 1& K+\gamma^{-1} I_{m}\end{pmatrix}
\begin{pmatrix}b\\ \boldsymbol \alpha\end{pmatrix}=\begin{pmatrix}0\\ \boldsymbol y \end{pmatrix},
\end{equation}
where $\boldsymbol 1=(1,\ldots,1)^{\mathrm{ T }}$, $I_{m}$ represents the $m\times m$ density matrix, $\gamma$ is a hyperparameter describing the ratio of the Lagrangian's components, $\boldsymbol \alpha=(\alpha_{1},\ldots,\alpha_{m})$ is the Lagrange multiplier, $\boldsymbol y=(y_{1},\ldots,y_{m})$ is the label of training set, $b$ is the offset of the hyperplane, and $K$ is an $m\times m$ kernel matrix. According to the Eq. (\ref{1}), for query data $\boldsymbol x\in \mathds{R}^{n}$, the classifier can be determined by the following function:
\begin{equation}\label{2}
f(x)=sign(\sum_{k=1}^{m}\alpha_{k}\boldsymbol x_{k}^{\mathrm{ T }}\boldsymbol x+b).
\end{equation}

In the quantum case, Eq. (\ref{1}) can be compactly rewritten as $F|b, \boldsymbol \alpha\rangle=|0, \boldsymbol y\rangle$.
When the
matrix $F$ is well conditioned and has sparsity polylogarithmic
in the dimension, the quantum state $|b,\boldsymbol \alpha\rangle=F^{-1}|0, \boldsymbol y\rangle$ is produced by using the HHL algorithm \cite{Aram2009Quantum}.
The desired LS-QSVM parameters can be represented as follows:
\begin{equation}\label{3}
|b,\boldsymbol \alpha\rangle=\frac{1}{\sqrt{C}}(b|0\rangle+\sum_{i=1}^{M}\alpha_{k}|k\rangle),
\end{equation}
where $C=b^{2}+\sum_{k=1}^{M}\alpha_{k}^{2}$.

For a classification task, the training data oracle is constructed as follows:

\begin{equation}\label{11}
  |\tilde{\boldsymbol\mu}\rangle=\frac{1}{\sqrt{N_{\tilde{\boldsymbol\mu}}}}(b|0\rangle|0\rangle+
  \sum_{k=1}^{m}\alpha_{k}|\boldsymbol x_{k}||k\rangle|\boldsymbol x_{k}\rangle),
\end{equation}
with $N_{\tilde{\mu}}=b^{2}+\sum_{k=1}^{m}\alpha_{k}^{2}|\boldsymbol x_{k}|^{2}$. In addition, the query state is constructed as follows:
\begin{equation}\label{12}
  |\tilde{\boldsymbol z}\rangle=\frac{1}{\sqrt{N_{\tilde{\boldsymbol z}}}}(|0\rangle|0\rangle+\sum_{k=1}^{m}|\boldsymbol x||k\rangle|\boldsymbol x\rangle|0\rangle),
\end{equation}
with $N_{\tilde{\boldsymbol z}}=m|\boldsymbol x|^{2}+1$.
For the
physical implementation of the inner product, Swap Test
can be utilized to project the ancillary qubit in state
$1/\sqrt{2}(|0\rangle|\tilde{\boldsymbol\mu}\rangle+|1\rangle|\tilde{\boldsymbol z}\rangle)$ to state $1/\sqrt{2}(|0\rangle-|1\rangle)$. Then the success probability is $P=1/2(1-\langle\tilde{\boldsymbol\mu}|\tilde{\boldsymbol z}\rangle)$, which can be used to
determine the label of $|\boldsymbol x\rangle$ according to the sign of $1/2-P$.

However, LS-QSVM only supports the binary classification.
Bishwas $et~al.$ handled the quantum multiclass classification problem by using all-pair technique \cite{Bishwas2018An} and one-against-all approach \cite{Bishwas2016Big}, respectively.
Hou  $et~al.$ \cite{Hou2020Quantum} presented semi-supervised SVM, which exhibited a quadratic speed-up over classical algorithm.
Feng $et~al.$  \cite{Feng2019Quantum} present an improved
QSVM model, exponentially improving the dependence on precision while keeping essentially the same dependence on other parameters. Besides,
the SVM with the quantum kernel algorithm was also proposed to solve classification problems \cite{Bishwas2020Gaussian,Schuld2018Quantum}.
\subsection{Quantum amplitude estimation}

Quantum amplitude estimation (QAE) algorithm \cite{Brassard2000Quantum} is a fundamental quantum algorithm that allows a quantum computer to
estimate the amplitude $U|0\rangle$ for a quantum circuit $U$.
QAE \cite{Brassard2000Quantum, Vazquez2021Efficient, Grinko2021Iterative} has attracted much attention as a fundamental subroutine of a wide range of application-oriented quantum algorithms, such as the Monte Carlo integration \cite{Montanaro2015Quantum, Rebentrost2018Quantum, Miyamoto2020Reduction} and machine learning tasks \cite{Miyahara2019Quantum, Kerenidis2019q}.
QAE was first introduced by Brassard $et~al.$ \cite{Brassard2000Quantum}, in which a unitary operator $\mathcal{A}_{1}$ acts on an initial state $|0\rangle^{\otimes p}$ and an ancillary $|0\rangle$ such that:
\begin{equation}\label{13}
\mathcal{A}_{1}|0\rangle^{\otimes p}|0\rangle=\sqrt{a_{1}}|\psi_{1}\rangle|1\rangle+\sqrt{1-a_{1}}|\psi_{0}\rangle|0\rangle.
\end{equation}
Let $|\varphi_{1}\rangle=\sqrt{a_{1}}|\psi_{1}\rangle|1\rangle$,  $|\varphi_{0}\rangle=\sqrt{1-a_{1}}|\psi_{0}\rangle|0\rangle$, and $|\Phi\rangle=|0\rangle^{\otimes p}|0\rangle$. Then
\begin{equation}\label{14}
|\varphi\rangle=\mathcal{A}_{1}|\Phi\rangle=|\varphi_{1}\rangle+|\varphi_{0}\rangle.
\end{equation}
Following \cite{Brassard2000Quantum},
$|\varphi_{1}\rangle$ and $|\varphi_{0}\rangle$ are respectively
the $p$-qubit normalized good and bad states.
QAE is to estimate the probability that the measuring
$|\varphi\rangle$ yields a good state, i.e., $a_{1}=\langle\varphi_{1}|\varphi_{1}\rangle$.
The probability to measure the good state can be amplified by applying the following unitary operator:
\begin{equation}\label{15}
\mathcal{Q}_{1}=-\mathcal{A}_{1}\mathcal{S}_{0}\mathcal{A}_{1}^{-1}\mathcal{S}_{\varphi_{1}},
\end{equation}
where $\mathcal{A}_{1}^{-1}$ is the inverse of $\mathcal{A}_{1}$, $\mathcal{S}_{0}=I-2|0\rangle\langle0|$, and $\mathcal{S}_{\varphi_{1}}=I-2|\varphi_{1}\rangle\langle\varphi_{1}|$.

In QAE, it uses $l$ ancillary qubits initialized in equal superposition to represent the final
result. Then it
applies the operator $\mathcal{Q}_{1}$ controlled by the
ancillary qubits. Eventually, it performs an inverse quantum fourier transformation ($FT^{+}$) on the ancillary qubits before they are measured. Subsequently, the measured integer $y\in{0,\ldots, L-1}$ ($L=2^{l}$) is mapped
to an angle $\tilde{\theta}_{a_{1}}=y\pi/L$. Thereafter, the resulting estimate of $a_{1}$ is
defined as $\tilde{a}_{1}=sin^{2}\tilde{\theta}_{a_{1}}$.

\section{The generalized quantum amplitude estimation}\label{S3}
In QAE \cite{Brassard2000Quantum}, it is assumed that the initial state of the interest problem is $|0\rangle^{\otimes p}$. 
In this section, we describe another method for the case where the
initial state is arbitrary state $|\Phi\rangle$. Our method is the generalization of the quantum amplitude estimation (GQAE).

Let $\mathcal{A}$ be a quantum algorithm that acts on the
initial state $|\Phi_{1}\rangle$ and an ancillary qubit $|0\rangle$ such that:
\begin{equation}\label{15.9}
\mathcal{A}|\Phi_{1}\rangle|0\rangle=|\Psi_{1}\rangle+|\Psi_{0}\rangle.
\end{equation}
Let $|\Phi\rangle=|\Phi_{1}\rangle|0\rangle$, we have
\begin{equation}\label{16}
|\Psi\rangle=\mathcal{A}|\Phi\rangle=|\Psi_{1}\rangle+|\Psi_{0}\rangle.
\end{equation}

We will present how to estimate the probability $a$ when measuring $|\Psi\rangle$ produces a good state, where $a=\langle\Psi_{1}|\Psi_{1}\rangle$.

The amplification amplitude operation Eq. (\ref{17}) is hereby important for estimation.
\begin{equation}\label{17}
\mathcal{Q}=-\mathcal{A}\mathcal{S}_{\Phi}\mathcal{A}^{-1}\mathcal{S}_{\chi},
\end{equation}
where $\mathcal{S}_{\chi}=I-2|\Psi_{1}\rangle\langle\Psi_{1}|$ conditionally changes the sign of the amplitudes of the good states,
\begin{equation}\label{18}
|x\rangle\longmapsto
\begin{cases}
-|x\rangle& \text{ if } x\text{~is~good~state}, \\
|x\rangle& \text{ if } x \text{~is~bad~state},
\end{cases}
\end{equation}
and the operator $\mathcal{S}_{\Phi}=I-2|\Phi\rangle\langle\Phi|$ changes the sign of the amplitude if and only if the
state is the initial state $|\Phi\rangle$.

In particular, it is the QAE proposed in \cite{Brassard2000Quantum},
if $|\Phi\rangle=|0\rangle^{\otimes s}|0\rangle$, where $s$ is the number of qubits used to represent $|\Phi_{1}\rangle$.

We will show that the quantum state $|\Psi\rangle$ can still be written as a linear combination of $|\Psi_{1}\rangle$ and $|\Psi_{0}\rangle$ after performing $j$ times of operator $\mathcal{Q}$.

\begin{lemma}\label{L3.1}
We have that
\begin{eqnarray}\label{19}
\mathcal{Q}|\Psi_{1}\rangle &=& (1-2a)|\Psi_{1}\rangle-2a|\Psi_{0}\rangle, \\
\mathcal{Q}|\Psi_{0}\rangle &=& 2(1-a)|\Psi_{1}\rangle+(1-2a)|\Psi_{0}\rangle,
\end{eqnarray}
where $a=\langle\Psi_{1}|\Psi_{1}\rangle$.
\end{lemma}

\begin{proof}
First consider the action of the operator $\mathcal{Q}$ on the
vector $|\Psi_{1}\rangle $:
\begin{eqnarray}\label{20}
\nonumber &&\mathcal{Q}|\Psi_{1}\rangle\\
\nonumber &=& -\mathcal{A}\mathcal{S}_{\Phi}\mathcal{A}^{-1}\mathcal{S}_{\chi} |\Psi_{1}\rangle \\
\nonumber &=& \mathcal{A}\mathcal{S}_{\Phi}\mathcal{A}^{-1}|\Psi_{1}\rangle \\
\nonumber &=&\mathcal{A}(I-2|\Phi\rangle\langle\Phi|)(|\Phi\rangle-\mathcal{A}^{-1}|\Psi_{0}\rangle)\\
\nonumber  &=&-\mathcal{A}|\Phi\rangle-\mathcal{A}\mathcal{A}^{-1}|\Psi_{0}\rangle+2\mathcal{A}|\Phi\rangle\langle\Phi|\mathcal{A}^{-1}|\Psi_{0}\rangle\\
  \nonumber&=&-(|\Psi_{1}\rangle+|\Psi_{0}\rangle)-|\Psi_{0}\rangle+2
  (|\Psi_{1}\rangle+|\Psi_{0}\rangle)
[(\langle\Psi_{1}|\\
\nonumber &&+\langle\Psi_{0}|)
  |\Psi_{0}\rangle]\\
  \nonumber&=&-|\Psi_{1}\rangle-2|\Psi_{0}\rangle+2(1-a)
  |\Psi_{1}\rangle+2(1-a)|\Psi_{0}\rangle\\
&=&(1-2a)|\Psi_{1}\rangle-2a|\Psi_{0}\rangle,
\end{eqnarray}
where $a=\langle\Psi_{1}|\Psi_{1}\rangle$.

Next consider the action of the operator $\mathcal{Q}$ on the
vector $|\Psi_{0}\rangle $:
\begin{eqnarray}\label{21}
\nonumber &&\mathcal{Q}|\Psi_{0}\rangle\\
\nonumber &=&-\mathcal{A}\mathcal{S}_{\Phi}\mathcal{A}^{-1}\mathcal{S}_{\chi}|\Psi_{0}\rangle \\
\nonumber &=&-\mathcal{A}\mathcal{S}_{\Phi}\mathcal{A}^{-1}|\Psi_{0}\rangle \\
\nonumber&=&-\mathcal{A}(I-2|\Phi\rangle\langle\Phi|)(|\Phi\rangle-\mathcal{A}^{-1}|\Psi_{1}\rangle)\\
\nonumber&=&\mathcal{A}|\Phi\rangle+\mathcal{A}\mathcal{A}^{-1}|\Psi_{1}\rangle-2\mathcal{A}|\Phi\rangle\langle\Phi|\mathcal{A}^{-1}|\Psi_{1}\rangle\\
         &=&2(1-a)|\Psi_{1}\rangle+(1-2a)|\Psi_{0}\rangle.
\end{eqnarray}
\end{proof}

Following the previous study \cite{Brassard2000Quantum}, the subspace $\mathcal{H}_{\Psi}$ has an orthonormal basis
consisting of two eigenvectors of $\mathcal{Q}$ according to Theorem \ref{The1}.

\begin{theorem} \label{The1}
For any initial quantum state $|\Phi\rangle$,
$|\Psi\rangle=\mathcal{A}|\Phi\rangle$ can be expressed as follows:
\begin{equation}\label{22}
|\Psi\rangle=\frac{-i}{\sqrt{2}}(e^{i\theta_{a}}|\Psi_{+}\rangle-
e^{-i\theta_{a}}|\Psi_{-}\rangle),
\end{equation}
where $|\Psi_{\pm}\rangle=\frac{1}{\sqrt{2}}(\frac{1}{\sqrt{a}}
|\Psi_{1}\rangle\pm\frac{i}{\sqrt{1-a}}|\Psi_{0}\rangle)$
represent two eigenvectors of $\mathcal{Q}$. The corresponding eigenvalues are
$\lambda_{\pm}=e^{\pm i2\theta_{a}}$, where $i=\sqrt{-1}$ denotes the principal square root of $-1$, and the angle $\theta_{a}$ is defined so that $sin^{2}(\theta_{a})=a$, and $0\leq\theta_{a}\leq\pi/2$.
\end{theorem}

\begin{proof}
The angle $\theta_{a}$ is defined so that $sin^{2}(\theta_{a})=a$. At first, we prove that $|\Psi\rangle$ can be represented as the linear combination of $|\Psi_{+}\rangle$ and $|\Psi_{-}\rangle$:
\begin{eqnarray}
\nonumber && |\Psi\rangle \\
\nonumber&=& |\Psi_{1}\rangle+|\Psi_{0}\rangle\\
\nonumber &=& \frac{2\sqrt{a}}{2\sqrt{a}}|\Psi_{1}\rangle+\frac{2\sqrt{1-a}}{2\sqrt{1-a}}|\Psi_{0}\rangle\\
\nonumber &=& \frac{-i}{2\sqrt{a}}(2isin\theta_{a})|\Psi_{1}\rangle+
\frac{1}{2\sqrt{1-a}}(2cos\theta_{a})|\Psi_{0}\rangle\\
\nonumber &=& \frac{-i}{2\sqrt{a}}(cos\theta_{a}+isin\theta_{a}-cos\theta_{a}+isin\theta_{a})
|\Psi_{1}\rangle\\
\nonumber&&+\frac{1}{2\sqrt{1-a}}(cos\theta_{a}+isin\theta_{a}+cos\theta_{a}-isin\theta_{a})|\Psi_{1}\rangle\\ \nonumber&=&\frac{-i}{\sqrt{2}}[(cos\theta_{a}+isin\theta_{a})(\frac{1}{\sqrt{2a}}|\Psi_{1}\rangle+
\frac{i}{\sqrt{2(1-a)}}|\Psi_{0}\rangle)\\ \nonumber&&-(cos\theta_{a}-isin\theta_{a})(\frac{1}{\sqrt{2a}}|\Psi_{1}\rangle-
\frac{i}{\sqrt{2(1-a)}}|\Psi_{0}\rangle)].
\end{eqnarray}
Let $|\Psi_{\pm}\rangle=\frac{1}{\sqrt{2}}(\frac{1}{\sqrt{a}}
|\Psi_{1}\rangle\pm\frac{i}{\sqrt{1-a}}|\Psi_{0}\rangle)$,
we have
\begin{equation}
|\Psi\rangle =\frac{-i}{\sqrt{2}}(e^{i\theta_{a}}|\Psi_{+}\rangle-
e^{-i\theta_{a}}|\Psi_{-}\rangle).
\end{equation}

Next, we prove that $|\Psi_{+}\rangle$ is the eigenvector of $\mathcal{Q}$. The corresponding eigenvalues are
$\lambda_{+}=e^{i2\theta_{a}}$.
\begin{eqnarray}\label{23}
\nonumber &&\mathcal{Q}|\Psi_{+}\rangle\\
\nonumber&=&\mathcal{Q} (\frac{1}{\sqrt{2a}}|\Psi_{1}\rangle+
\frac{i}{\sqrt{2(1-a)}}|\Psi_{0}\rangle)\\
 \nonumber &=& \frac{1}{\sqrt{2a}}\mathcal{Q} |\Psi_{1}\rangle+\frac{i}{\sqrt{2(1-a)}}\mathcal{Q} |\Psi_{0}\rangle\\
 \nonumber
 &=&[\frac{1-2a}{\sqrt{2a}}+\frac{2i(1-a)}{\sqrt{2(1-a)}}] |\Psi_{1}\rangle+[\frac{-2a}{\sqrt{2a}}+\frac{i(1-2a)}{\sqrt{2(1-a)}}] |\Psi_{0}\rangle\\
 \nonumber \nonumber&=&(1-2a+2i\sqrt{a}\sqrt{1-a})(\frac{1}{\sqrt{2a}}|\Psi_{1}\rangle+\frac{i}{\sqrt{2(1-a)}}|\Psi_{0}\rangle)\\
 &=&e^{i2\theta_{a}}|\Psi_{+}\rangle,
\end{eqnarray}
where the third equal sign is based on Lemma \ref{L3.1}.

The same method can be used to prove that $\mathcal{Q}|\Psi_{0}\rangle=e^{-i2\theta_{a}}|\Psi_{1}\rangle$. The result is confirmed.
\end{proof}

After $j$ applications of operator $\mathcal{Q}$, the state is
\begin{equation}\label{24}
\mathcal{Q}^{j}|\Psi\rangle=\frac{-i}{\sqrt{2}}(e^{(2j+1)i\theta_{a}}|\Psi_{+}\rangle-
e^{-(2j+1)i\theta_{a}}|\Psi_{-}\rangle).
\end{equation}

We can use the same method as introduced in \cite{Brassard2000Quantum} to estimate $a$ (Algorithm 1 in Table \ref{tab1}). The resulting estimate of $a$ is
defined as $\tilde{a}$. The quantum circuit to implement
this algorithm is depicted in Fig. \ref{fig3}.
\begin{table}
\caption{The generalized quantum amplitude estimation}
\setlength{\tabcolsep}{0.1pt}
\begin{tabular}{p{240pt}}
\hline
$\mathbf{Algorithm~1:}$ The generalized quantum amplitude estimation\\
\hline
$\mathbf{1}$. Prepare the quantum state $|0\rangle^{\otimes h}|\Phi_{1}\rangle|0\rangle$.\\

$\mathbf{2}$. Apply Hadamard to the first register $|0\rangle^{\bigotimes h}$, where $h$ is an integer that relates to the precision.\\

$\mathbf{3}$. Apply $\mathcal{A}$ to the second register $|\Phi_{1}\rangle|0\rangle$.\\

$\mathbf{4}$. Apply $\mathcal{Q}^{2^{j}}$ to the second register controlled by the first register.\\

$\mathbf{5}$. Apply Fourier transformation in inverse to the first register.\\
$\mathbf{6}$. Measure the first register and denote the outcome $|y\rangle$.\\

$\mathbf{7}$. Output $\tilde{a}=sin^{2}(\pi\frac{y}{2^{h}})$.
 \\
 \hline
\end{tabular}
\label{tab1}
\end{table}

\begin{figure}
  \centering
  \includegraphics[width=8.5cm]{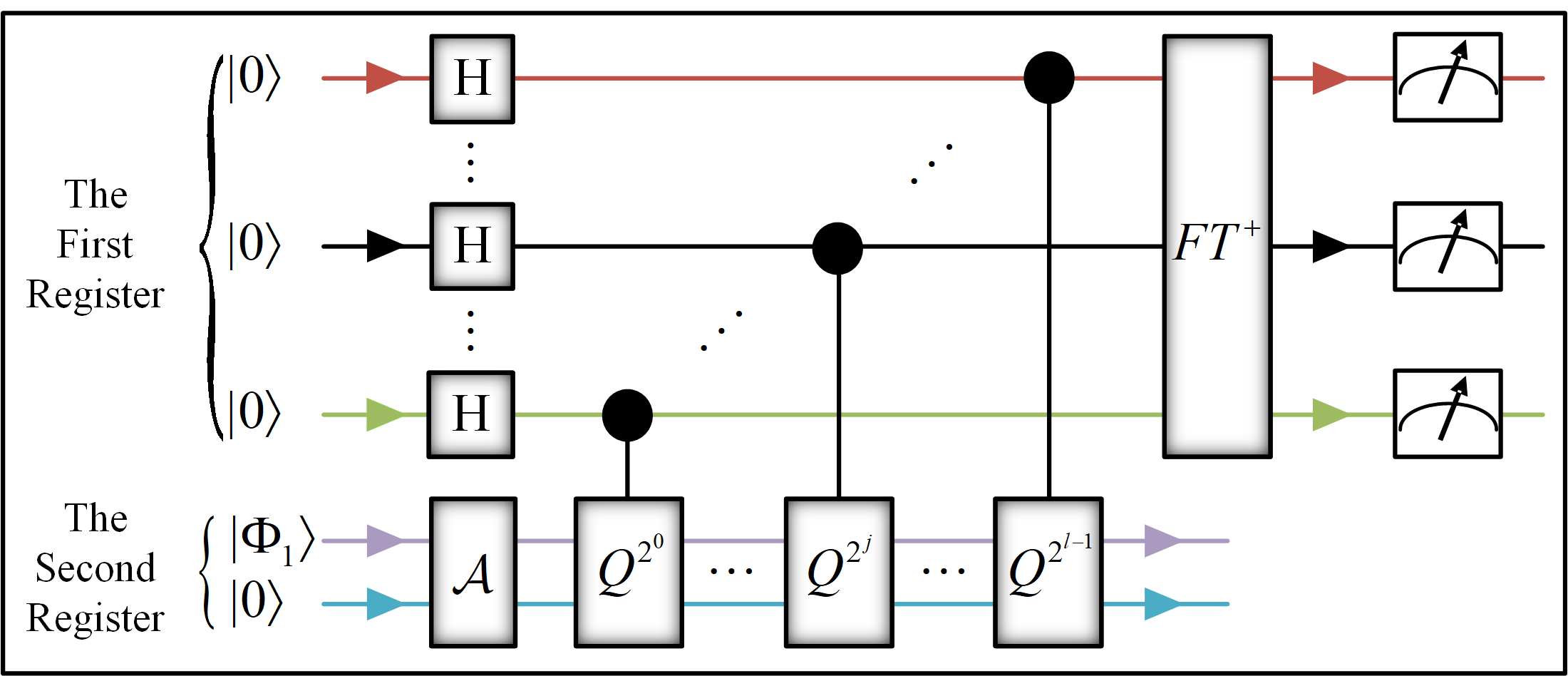} \\
  \caption{The quantum circuit of the generalized quantum amplitude estimation}
  \label{fig3}
\end{figure}

\section{The quantum support vector machine based on GQAE}\label{S4}
After introducing the tool we cast our gaze over the QSVM based on the GQAE (AE-QSVM), striving to address the problem arising in training and classification.
At first, the training progress of the QSVM is realized based on quantum singular value decomposition. Then a query state was classified using GQAE to calculate the inner product.

\subsection{Quantum support vector machine training}\label{S4.1}
In LS-QSVM, the input matrix $F$ is not expected to be sparse and well structured, thus restricting the applications of HHL algorithm.
An important method for solving linear equations is singular value decomposition.
This paper use the quantum singular value decomposition to train AE-QSVM \cite{Kerenidis2020Quantum}.

Singular value is equal to eigenvalue for the square matrix $F$.
Therefore, the singular value decomposition of matrix $F$ is written as $F=\sum_{i}\lambda_{i}\nu_{i}\nu_{i}^{\mathrm{ T }}$, where $\lambda_{i}\geq0$ are the eigenvalues and $\nu_{i}$ are the corresponding eigenvectors.
If the matrix  $F$ is singular, the
solution $F^{-1}|0, \boldsymbol y\rangle$ can be achieved for the Moore-Penrose pseudoinverse, that is, the nonzero eigenvalues are only inverted.

As in Algorithm $2$ (Table \ref{tab2}), the quantum linear systems $F|b, \boldsymbol \alpha\rangle=|0, \boldsymbol y\rangle$ is solved by using the singular value decomposition algorithm \cite{Kerenidis2020Quantum}. We obtain the parameters of AE-QSVM:
\begin{equation}\label{24.9}
|b,\boldsymbol\alpha\rangle=\sum_{i}\beta_{i}\frac{1}{\bar{\lambda}_{i}}|\nu_{i}\rangle.
\end{equation}

In the basis of training set labels, the expansion coefficients of the ultimate state are the desired support vector machine parameters:
\begin{equation}\label{25}
|b,\boldsymbol\alpha\rangle=\frac{1}{\sqrt{C}}(b|0\rangle+\sum_{i=1}^{M}\alpha_{k}|k\rangle),
\end{equation}
where $C=b^{2}+\sum_{k=1}^{M}\alpha_{k}^{2}$.

\begin{table}
\caption{Solve the quantum linear systems by using the improved singular
value decomposition algorithm}
\label{table}
\setlength{\tabcolsep}{3pt}
\begin{tabular}{p{240pt}}
\hline
$\mathbf{Algorithm~2:}$ Solve the quantum linear systems using the improved singular
value decomposition algorithm \cite{Kerenidis2020Quantum} \\
\hline
Require: Matrix $F$ stored in the data structure, such that eigenvalues of $F$ lie in $[1/\kappa,1]$. Input state $|0,y\rangle
=\sum_{i}\beta_{i}|\nu_{i}\rangle$
.\\

$\mathbf{1}$. Perform singular value decomposition with precision $\varepsilon_{1}$ for $F$ on
$|0,y\rangle$ to obtain $\sum_{i}\beta_{i}|\nu_{i}\rangle|\bar{\lambda}_{i}\rangle$
.
\\

$\mathbf{2}$. Perform a conditional rotation and uncompute the SVE
register to obtain the state:$\sum_{i}\beta_{i}|\nu_{i}\rangle
(\frac{1}{\bar{\lambda}_{i}}|0\rangle+\gamma|1\rangle)$.\\

$\mathbf{3}$. Perform the amplitude amplification to obtain
$|b,\alpha\rangle=\sum_{i}\beta_{i}\frac{1}{\bar{\lambda}_{i}}|\nu_{i}\rangle$.
 \\
 \hline
\end{tabular}
\label{tab2}
\end{table}

\subsection{Classification based on generalized quantum amplitude estimation}
We have implement the training procedure of quantum support vector machine and would like to classify a query state $|\boldsymbol x\rangle$.

Following Ref. \cite{Rebentrost2014Quantum},
the training data oracle is constructed as follows by adding a register to the state $|b,\boldsymbol\alpha\rangle$:
\begin{equation}\label{26}
  |\tilde{\boldsymbol\mu}\rangle=\frac{1}{\sqrt{N_{\tilde{\boldsymbol\mu}}}}(
  b|0\rangle|0\rangle+\sum_{k=1}^{M}\alpha_{k}|\boldsymbol x_{k}||k\rangle|\boldsymbol x_{k}\rangle),
\end{equation}
with $N_{\tilde{\mu}}=b^{2}+\sum_{k=1}^{M}\alpha_{k}^{2}|\boldsymbol x_{k}|^{2}$. The query state is constructed:
\begin{equation}\label{27}
  |\tilde{\boldsymbol x}\rangle=\frac{1}{\sqrt{N_{\tilde{\boldsymbol z}}}}(|0\rangle|0\rangle+\sum_{k=1}^{M}|\boldsymbol x||k\rangle|\boldsymbol x\rangle),
\end{equation}
with $N_{\tilde{\boldsymbol z}}=M|\boldsymbol x|^{2}$.

For the classification, we perform a quantum inner estimation for $|\tilde{\boldsymbol\mu}\rangle$ and $|\tilde{\boldsymbol x}\rangle$ based on GQAE for the following reason:
\begin{equation}\label{28}
\langle\tilde{\boldsymbol\mu}|\tilde{\boldsymbol x}\rangle=\frac{1}{\sqrt{N_{\tilde{\boldsymbol\mu}}}\sqrt{N_{\tilde{\boldsymbol z}}}}(b+\sum_{k=1}^{M}\alpha_{k}|\boldsymbol x_{k}||\boldsymbol x|\langle\boldsymbol x_{k}|\boldsymbol x\rangle).
\end{equation}

Following the previous studies \cite{Lin2020Quantum,Li2019Quantum,Hou2020Quantum}, two ancillary qubits can be used to construct entangled states
\begin{equation}\label{29}
|\phi_{0}\rangle=\frac{1}{\sqrt{2}}(|0\rangle|\tilde{\boldsymbol\mu}\rangle-|1\rangle|\tilde{\boldsymbol x}\rangle)|0\rangle.
\end{equation}

Firstly, a Hadamard gate is acted on the first register (denoted as $H_{1}$) to obtain the following state:
\begin{eqnarray}\label{30}
 \nonumber &&|\phi_{1}\rangle\\
 \nonumber &=& H_{1}|\phi_{0}\rangle\\
\nonumber&=& \frac{1}{\sqrt{2}}(\frac{|0\rangle+|1\rangle}{\sqrt{2}}|\tilde{\boldsymbol\mu}\rangle-
   \frac{|0\rangle-|1\rangle}{\sqrt{2}}|\tilde{\boldsymbol x}\rangle)|0\rangle \\
&=& \frac{1}{2}(|0\rangle|\tilde{\boldsymbol\mu}\rangle|0\rangle+|1\rangle|\tilde{\boldsymbol\mu}\rangle|0\rangle
   -|0\rangle|\tilde{\boldsymbol x}\rangle|0\rangle+|1\rangle|\tilde{\boldsymbol x}\rangle|0\rangle).
\end{eqnarray}

Then, the Pauli-$X$ gate is applied to the third register controlled by the first register with $|0\rangle$ (denoted as $C\-/X$). That is,  $|0\rangle|i\rangle|j\rangle\rightarrow|0\rangle|i\rangle|j\oplus1\rangle$ and $|1\rangle|i\rangle|j\rangle\rightarrow|1\rangle|i\rangle|j\rangle$, where $\oplus$ is modulo $2$ addition.
\begin{eqnarray}\label{31}
\nonumber &&|\phi_{2}\rangle\\
\nonumber
&=&(C\-/X)|\phi_{1}\rangle\\
\nonumber
   &=&\frac{1}{2}(|0\rangle|\tilde{\boldsymbol\mu}\rangle|1\rangle+|1\rangle|\tilde{\boldsymbol\mu}\rangle|0\rangle
   -|0\rangle|\tilde{\boldsymbol x}\rangle|1\rangle+|1\rangle|\tilde{\boldsymbol x}\rangle|0\rangle)\\
 &=& \frac{1}{2}(|0\rangle|\tilde{\boldsymbol\mu}\rangle-|0\rangle|\tilde{\boldsymbol x}\rangle)|1\rangle+
       \frac{1}{2}(|1\rangle|\tilde{\boldsymbol\mu}\rangle+|1\rangle|\tilde{\boldsymbol x}\rangle)|0\rangle.
\end{eqnarray}

Let $\mathcal{A}_{1}=(H_{1})(C\-/X)$ ,$\psi_{1}=\frac{1}{2}(|0\rangle|\tilde{\boldsymbol\mu}\rangle-|0\rangle|\tilde{\boldsymbol x}\rangle)|1\rangle$, $\psi_{0}=\frac{1}{2}(|1\rangle|\tilde{\boldsymbol\mu}\rangle+|1\rangle|\tilde{\boldsymbol x}\rangle)|0\rangle$.
The problem of interest is given by an operator $\mathcal{A}_{1}$ acting on $|\phi_{0}\rangle$ such that
\begin{equation}\label{32}
\mathcal{A}_{1} : |\phi_{0}\rangle\longrightarrow |\psi_{1}\rangle+|\psi_{0}\rangle.
\end{equation}

Then we can use GQAE to calculate $\langle\psi_{1}|\psi_{1}\rangle$.
Therefore, the inner product of $|\tilde{\boldsymbol\mu}\rangle$ and $|\tilde{\boldsymbol x}\rangle$ can be obtained for the following equation:
\begin{equation}\label{33}
\langle\psi_{1}|\psi_{1}\rangle=\frac{1}{2}(1-\langle \tilde{\boldsymbol\mu}|\tilde{\boldsymbol x}\rangle).
\end{equation}
When we estimate that the inner product of $\langle\psi_{1}|\psi_{1}\rangle$ is greater than $\frac{1}{2}$ through GQAE, $|\boldsymbol x\rangle$ belongs to the $+1$; otherwise, $-1$.

Obviously, AE-QSVM is executed one time to obtain the classification result with an accuracy $\varepsilon$ when using GQAE to classify new samples, where $\varepsilon$ depends on the number of auxiliary qubits in the GQAE algorithm.
The following will compare the algorithms LS-QSVM and AE-QSVM in terms of qubit consumption and complexity.

\subsection{Space complexity analysis of AE-QSVM }
The space complexity of AE-QSVM is analyzed in this subsection. That is we want to analyzed the qubit consumption.
For AE-QSVM, the qubit consumptions in training process is $3+\lceil log(m+1)\rceil+k+\lceil log(2+\frac{1}{2\varepsilon})\rceil+1$, where $k$ represents that the eigenvalues of $F$ are approximated to an accuracy of $2^{-k}$ \cite{Nielsen2000Quantum, Kerenidis2020Quantum}.

Estimating the inner product in AE-QSVM, $h$ ancillary qubit is needed to have an accuracy of $\varepsilon$ \cite{Brassard2000Quantum}.
\begin{equation}\label{33.2}
\varepsilon=\frac{2\pi\sqrt{a(1-a)}}{H}+\frac{\pi^{2}}{H^{2}},
\end{equation}
where $H=2^{h}$.
Solving Eq. (\ref{33.2}) and considering $H>0$, we have
\begin{equation}\label{33.3}
H=\frac{\pi(\sqrt{a(1-a)}+\sqrt{a(1-a)+\varepsilon})}{\varepsilon}.
\end{equation}
Therefore,
\begin{equation}\label{33.3}
h\leq\lceil log\frac{\pi+\sqrt{3}\pi}{\varepsilon}\rceil.
\end{equation}
Consequently, the qubit consumption of AE-QSVM is
\begin{equation}\label{33.4}
3+\lceil log(m+1)\rceil+\lceil log\varepsilon\rceil+\lceil log(2+\frac{1}{2\varepsilon})\rceil+1+\lceil log\frac{\pi+\sqrt{3}\pi}{\varepsilon}\rceil.
\end{equation}

Besides, for LS-QSVM, the qubit consumption in training process is also $3+\lceil log(m+1)\rceil+k+\lceil log(2+\frac{1}{2\varepsilon})\rceil+1$. We set the accuracy all to $\varepsilon$ in this paper; thus the qubit consumption of LS-QSVM is \begin{equation}\label{33.1}
T\times(3+\lceil log(m+1)\rceil+\lceil log\varepsilon\rceil+\lceil log(2+\frac{1}{2\varepsilon})\rceil+1),
\end{equation}
where $T$ represents the number of the iterations.

Let $Q=3+\lceil log(m+1)\rceil+\lceil log\varepsilon\rceil+\lceil log(2+\frac{1}{2\varepsilon})\rceil+1$, the space complexity of LS-QSVM and AE-QSVM respectively are $Q+\lceil log\frac{\pi+\sqrt{3}\pi}{\varepsilon}\rceil$  and $TQ$. That is to say, LS-QSVM have a higher space complexity when $T>1+\frac{\lceil log\frac{\pi+\sqrt{3}\pi}{\varepsilon}\rceil}{Q}$ compared with AE-QSVM.

\subsection{Time complexity analysis of AE-QSVM }

We now analyze the time complexity for building AE-QSVM.
In the stage of the training process, the time is dominated by the quantum singular value decomposition \cite{Kerenidis2020Quantum}. We set $\kappa$ as the condition number (the largest eigenvalue divided by the smallest
eigenvalue), and only normalized eigenvalues $\lambda_{j}$ in the interval $1/\kappa\leq|\lambda_{j}|\leq1$ are taken into account. The kernel matrix $F^{0}$
with time $O(log(mn))$ \cite{Rebentrost2014Quantum} is prepared.
The running time of compute $e^{iF^{0}t_{0}}$ is $O(t_{0}^{2}\varepsilon^{-1})$, where $t_{0}$ is evolution time and $\varepsilon$ is the accuracy.
Therefore, the phase estimation costs time $O(t^{2}_{0}\varepsilon^{-1}log(mn))$ which presents the major time complexity. The probability of
getting $\lambda_{j}^{-1}$ determines iteration times of quantum singular value decomposition, and $O(\kappa^{2})$
iterations are needed to ensure the high success probability. However, only $O(\kappa)$ repetitions
are performed to obtain the same success probability by amplitude amplification; thus the time
complexity is $O(\kappa t_{0}^{2}\varepsilon^{-1}log(mn))$. For $t_{0}=O(\kappa\varepsilon^{-1})$, the complexity can be written as
$O(\kappa^{3}\varepsilon^{-3}log(mn))$.

In the stage of classifying a new sample, the complexity is dominated by quantum amplitude estimation. The complexity of performing unitary operator $\mathcal{Q}$ is $O(1)$ \cite{Nielsen2000Quantum}; thus the complexity of classification is $O(\kappa^{3}\varepsilon^{-3})$.

Putting all the time together, the time complexity of AE-QSVM is $O(\kappa^{3}\varepsilon^{-3}(log(mn)+1))$.

\section{Simulation}\label{S5}
The main work of this paper is to solve the restriction problem of parameter matrix and the large resource consumption caused by quantum collapse when implementing the Swap Test. Therefore, our experiment is divided into two parts.
\subsection{The simulation about the parameter matrix}
In this subsection, we give an example to illustrate the problem of the matrix consisted of an input sample.
In the LS-QSVM, the parameter matrix $F$ is a positive definite matrix or a positive semidefinite matrix.
When $F$ is a positive semidefinite matrix, the inverse does not exist. To illustrate the validity of the algorithm used in subsection \ref{S4.1}, consider the following matrix:
\begin{equation}\label{34}
F_{1}=
\begin{pmatrix}5&-1&3\\ -1&5&-3\\3&-3&3\end{pmatrix}.
\end{equation}

The eigenvalues of the matrix $F_{1}$ is $\lambda=0, 4, 9$,  and their corresponding eigenvectors is

\begin{eqnarray}
\nonumber \boldsymbol\nu_{1}&=&\begin{pmatrix} 881/2158&  -881/2158&-881/1079\end{pmatrix}^{\mathrm{ T }},\\
\nonumber\boldsymbol\nu_{2}&=&\begin{pmatrix}  985/1393&   985/1393&0\end{pmatrix}^{\mathrm{ T }},\\
\nonumber\boldsymbol\nu_{3}&=&\begin{pmatrix} -780/1351&  780/1351&-780/1351\end{pmatrix}^{\mathrm{ T }}.
\end{eqnarray}

The matrix $F_{1}$ is singular since the first eigenvalue $\lambda=0$, i.e., the eigenvalue decomposition can not be used to solve the quantum linear systems $F_{1}x=b$. According to the analysis in subsection \ref{S4.1}, we invert only the nonzero eigenvalues when the matrix $F_{1}$ is singular.

\subsection{The simulation about the quantum resource consumption}
Experiments are performed for the complete algorithm in this subsection. In theory, the high complexity in terms of space and time of LS-QSVM are caused by the Swap Test.
Therefore, we first conduct experiments on the IBM platform for Swap Test. Then LS-QSVM and AE-QSVM are compared in both space and time.

At first, we conduct the experiment on the IBM qiskit platform.
The quantum circuit of Swap Test for classification is shown in Fig. \ref{fig1}.
The probability of getting $|1\rangle$ is $P$ when we measure the ancillary qubit.
$P$ can be obtained to an accuracy $\varepsilon$ by iterating $O(P(1-P)/\varepsilon^{2})$ as shown in Ref. \cite{Rebentrost2014Quantum}.
The number of iterations is identical for probability $P$ and $1-P$ when achieving the same accuracy because of $0<P<1$.
Therefore, we let $P=0.1$, $P=0.3$, and $P=0.5$.
It can be seen from Fig. \ref{fig4} that the number of iterations increases as the accuracy increases. Specifically,
the number of iterations will sharply increase as the error increases when $0<\varepsilon<0.1$.
\begin{figure}
  \centering
  \includegraphics[width=8.5cm]{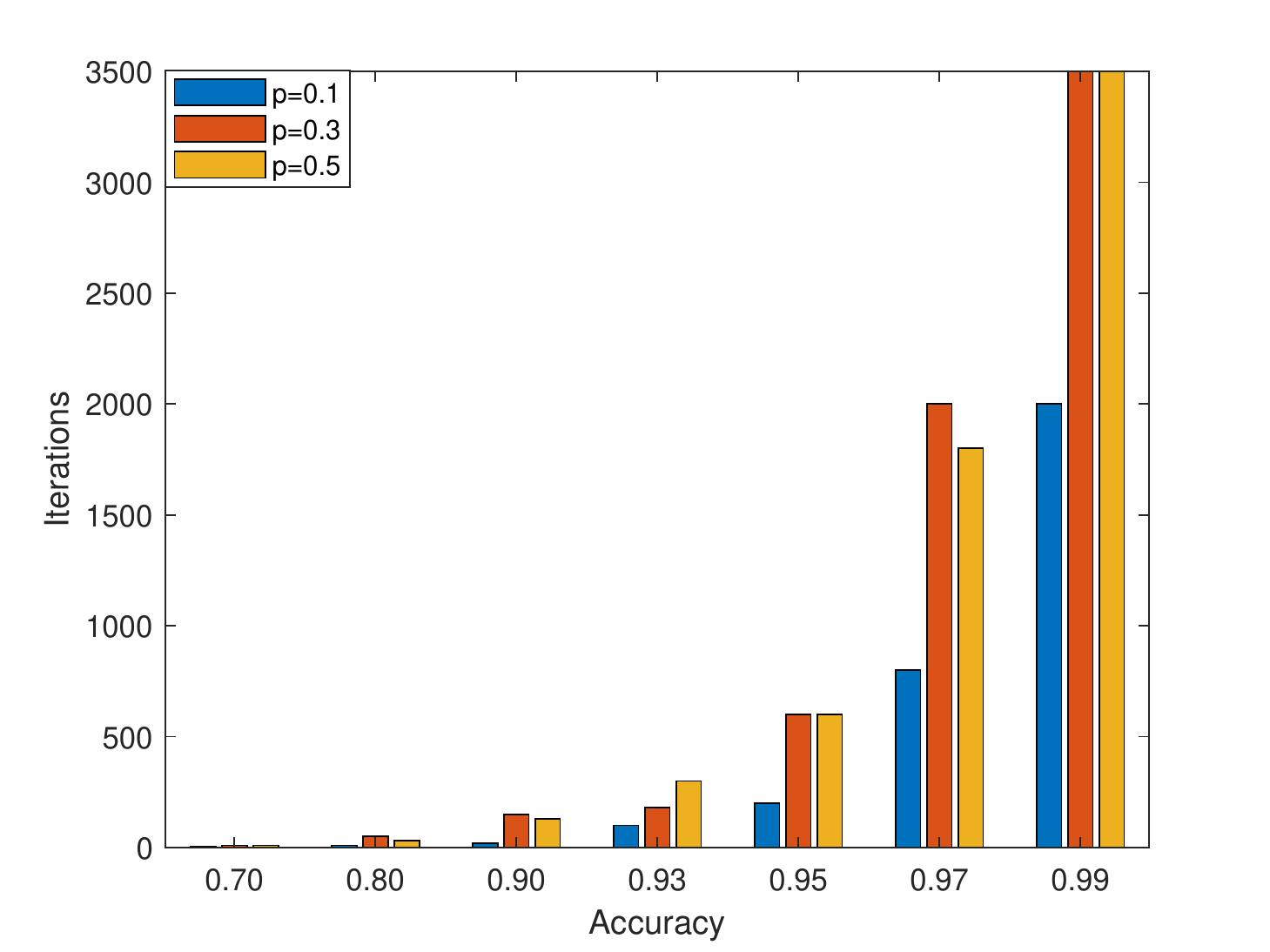}\\
  \caption{The number of iterations along error for Swap Test}
  \label{fig4}
\end{figure}

Then, the qubit consumption is presented based on the above result.  Fig. \ref{fig4} shows that a large number of iterations are required when using LS-QSVM for classification.
Besides, the quantum state of each iteration cannot be used for the next iteration.
Therefore, a large number of qubits are needed to be consumed.
AE-QSVM does not need iteration, and the number of qubits consumption is shown in Eq. (\ref{33.4}).
It can be seen from Fig. \ref{fig4} that the number of iterations required is different for the different $P$.  In this section, we take the average of three iterations for the same accuracy.
Table \ref{tab3} shows the number of qubits consumed by AE-QSVM  and LS-QSVM for different sample sizes ($m$).
It can be seen from Table \ref{tab3} that the number of qubits to be consumed by both algorithms increases as $m$ increases.
However, there are more qubits consumed by algorithm LS-QSVM  than that of algorithm AE-QSVM.
Especially, as the accuracy increases, the total amount of qubits consumed by algorithm LS-QSVM rapidly increases, while the total number of qubits consumed by algorithm AE-QSVM remains unchanged.
\begin{table*}
\centering
\caption{LS-QSVM compared with AE-QSVM about the number of qubit consumption}
\setlength{\tabcolsep}{5pt}
\begin{tabular}{|c|c|
c|c|c|c|c|c|c|}
\hline
\diagbox{Algorithm}{Qubits}{Accuracy}  &   &0.70    &0.80  &0.90    &0.93  &0.95 &0.97  &0.99\\
\hline
~~~~~~AE-QSVM &\multirow{2}* {$m$=10}&14& 14 & 14  & 14 &14   & 14   &14 \\
~~~~~~LS-QSVM &  ~                 &84& 300& 1000&1934&4667 &15334 &30000 \\
\hline
~~~~~~AE-QSVM &\multirow{2}* {$m$=100}&17 & 17 & 17  & 17 &17   & 17   &17 \\
~~~~~~LS-QSVM &  ~                  &109& 390& 1300&2513&6067 &19933 &39000 \\
\hline
~~~~~~AE-QSVM &\multirow{2}* {$m$=1000}&20 & 20 & 20  & 20 &20   & 20   &20 \\
~~~~~~LS-QSVM &  ~                   &134& 480& 1600&3034&7467 &34534 &48000 \\
\hline
~~~~~~AE-QSVM & \multirow{2}* {$m$=10000}&24 & 24 & 24  & 24 &24   & 24   &24 \\
~~~~~~LS-QSVM &  ~                     &167& 600& 2000&3867&9334 &30667 &60000 \\
\hline
\end{tabular}
\label{tab3}
\end{table*}

Finally, the complexity of AE-QSVM and LS-QSVM are compared.
It can be seen from the above experiment that extensive iterations are needed to achieve a high degree of accuracy for LS-QSVM.
Although the complexity of inner estimation is higher than that of the Swap Test, the measurement is not needed in AE-QSVM.
Therefore, we will compare the complexity of the two algorithms.
In order to compare LS-QSVM and AE-QSVM, we use expectation in mathematics to describe the complexity with an accuracy $\varepsilon$, i.e.,:
\begin{eqnarray}
\nonumber &&O(\int_{0}^{1}(\kappa^{3}\varepsilon^{-3}log(mn)+log(n))\frac{P(1-P)}
{\varepsilon^{2}}PdP) \\
&=&O(\frac{\kappa^{3}\varepsilon^{-3}log(mn)+log(n)}{12\varepsilon^{2}}).
\end{eqnarray}

The complexity of AE-QSVM is
\begin{equation}\label{38}
O(\kappa^{3}\varepsilon^{-3}(log(mn)+1))
\end{equation}
as analyzed in the above section.

Experiments are designed based on the six datasets: \emph{Tic~tac~toe} ($m=958$, $n=9$), \emph{Haberman} ($m=306$, $n=3$), \emph{Ionosphere} ($m=351,n=34$), \emph{Heart statlog} ($m=270$, $n=13$), \emph{Liver
disorders} ($m=345$, $n=6$), and \emph{Bupa} ($m=345,n=6$).
We first calculated the condition numbers of the six datasets and then calculated the complexity of the two methods according to Eq. (\ref{13}) and Eq. (\ref{14}) respectively.
It can be seen from Table \ref{tab4} that the complexity of AE-QLSM is significantly lower than that of LS-QSVM  when the accuracy is above $0.2$, and the complexity of the two algorithms differs little when the accuracy is $0.3$.
However, the algorithm is meaningless if the error is too large.
Therefore, our algorithm is superior than the LS-QSVM algorithm when the error is in the acceptable range.

\begin{sidewaystable}
\centering
\caption{LS-QSVM compared with AE-QSVM on six datasets}
\label{tab4}
\setlength{\tabcolsep}{5pt}
\begin{tabular}{|c|c|
c|c|c|c|c|c|c|}
\hline
\diagbox{Algorithm}{Complexity}{ Accuracy}  & Dataset   &0.99    &0.97  &0.95    &0.93  &0.90 &0.80  &0.70\\
\hline
AE-QSVM    &\multirow{2}* {\emph{Tic~tac~toe}} &    $1.32\times10^{24}$  &      $4.89\times10^{22}$ & $1.06\times10^{22}$  &   $3.85\times10^{21}$   &   $1.32\times10^{21}$    &   $1.65\times10^{20}$     & $4.89\times10^{19}$ \\
LS-QSVM     &  &   $9.91\times10^{28}$  &  $1.36\times10^{26}$ & $6.34\times10^{24}$ &  $8.43\times10^{23}$  &     $9.91\times10^{22}$
&   $1.55\times10^{21}$     & $1.36\times10^{20}$ \\
\hline
AE-QSVM     &\multirow{2}* {\emph{Haberman}} &       $3.65\times10^{27}$    &      $1.35\times10^{26}$& $2.92\times10^{25}$ &   $1.06\times10^{25}$   &   $3.65\times10^{24}$  &   $4.56\times10^{23}$   &   $1.35\times10^{23}$    \\
LS-QSVM    &    &      $2.67\times10^{32}$    &      $3.66\times10^{29}$ & $1.71\times10^{28}$  &   $2.67\times10^{26}$   &   $6.51\times10^{28}$ &   $4.17\times10^{24}$   &   $3.66\times10^{23}$     \\
\hline
AE-QSVM    &\multirow{2}* {\emph{Ionosphere}}  &     $1.46\times10^{20}$   &      $5.40\times10^{18}$ & $1.17\times10^{18}$ &   $4.25\times10^{17}$  &   $1.46\times10^{17}$ &   $1.82\times10^{16}$  &   $5.40\times10^{15}$  \\
LS-QSVM    &     &      $1.10\times10^{25}$    &      $1.51\times10^{22}$ & $7.03\times10^{20}$ &   $9.34\times10^{19}$  &   $1.10\times10^{19}$  &   $1.72\times10^{17}$  &   $1.51\times10^{16}$    \\
\hline
AE-QSVM      & \multirow{2}* {\emph{Heart statlog}}&     $6.39\times10^{31}$    &      $2.37\times10^{30}$ & $5.11\times10^{29}$  &   $1.86\times10^{29}$  &   $6.38\times10^{28}$    &   $7.97\times10^{27}$  &   $2.36\times10^{27}$  \\
LS-QSVM      &   &      $4.75\times10^{36}$    &      $6.51\times10^{33}$ & $3.03\times10^{32}$ &   $4.29\times10^{31}$  &   $4.74\times10^{30}$ &   $7.41\times10^{28}$  &   $6.50\times10^{27}$   \\
\hline
AE-QSVM     & \multirow{2}* {\emph{Liver
disorders}} &     $1.45\times10^{32}$    &      $5.39\times10^{30}$ & $1.16\times10^{30}$ &   $4.24\times10^{29}$  &   $1.45\times10^{29}$  &   $1.82\times10^{28}$  &   $5.39\times10^{27}$  \\
LS-QSVM    &     &      $1.07\times10^{37}$     &      $1.47\times10^{34}$& $6.88\times10^{32}$ &   $9.13\times10^{31}$   &   $1.07\times10^{31}$     &   $1.68\times10^{29}$  &   $1.47\times10^{28}$ \\
\hline
AE-QSVM      & \multirow{2}* {\emph{Bupa}} & $2.14\times10^{29}$   &      $7.92\times10^{27}$  & $1.71\times10^{27}$ & $6.23\times10^{26}$  &   $2.14\times10^{26}$   & $2.67\times10^{25}$  &   $7.92\times10^{24}$   \\
LS-QSVM        &  &      $1.58\times10^{34}$    &      $2.17\times10^{31}$ & $1.01\times10^{30}$ &   $1.34\times10^{29}$  &   $1.58\times10^{28}$    & $2.47\times10^{26}$  &   $2.17\times10^{25}$     \\
\hline
\end{tabular}
\end{sidewaystable}

\section{Discussion and conclusion}\label{S6}
The quantum support vector machine based on quantum amplitude estimation is introduced in this paper.
At first, the quantum singular value decomposition was used to train AE-QSVM which excludes the constraints of sparsity and is well structured for input matrix.
Then, we used the quantum amplitude estimation instead of Swap Test to realize the classification so that the support vector machine was not needed to be repetitively performed.
In AE-QSVM, auxiliary qubits are used to transform the probability amplitude into basis state; thus high accuracy can be achieved by adding the number of auxiliary qubits. Therefore, the algorithm AE-QSVM essentially reflects the superiority of quantum algorithms.

The results from this study can extended towards the following directions in future work:

 (1) AE-QSVM can be extended to the nonlinear classifier. One of the most powerful uses of support vector machines is to perform the nonlinear classification \cite{Cortes1995Suppor}.
One can use our proposed algorithm RN-QSVM to find a nonlinear classifier based on the Gaussian kernel \cite{Smola2000Generalized}:

\begin{equation}\label{12.1}
  (K(A,B))_{ij}=e^{-\delta\|A_{i}^{'}-B_{\cdot j}\|^{2}},~~~~~~~~~ i=1,\ldots, m, j=1,\ldots k
\end{equation}
where $A\in R^{m\times n}$ and $B\in R^{n\times l}$, $\delta$ is a positive constant, and the kernel $K(A,B)$ maps $R^{m\times n} \times R^{n\times l}$ into $R^{m\times l}$.
A quantum version of the Gaussian kernel was proposed by Bishwas $et~al.$ \cite{Bishwas2020Gaussian}.
The running time complexity of the quantum Gaussian kernel is significantly shorter compared with its classical version.

(2) The inner estimation based on amplitude estimation has low complexity compared with the Swap Test. Therefore, this generalized amplitude estimation can be used in other machine learning methods, for instance, neural networks and clustering.

\ifCLASSOPTIONcompsoc
  \section*{Acknowledgments}
\else
  \section*{Acknowledgment}
\fi

This work was supported by the National Key Research and Development Program of China under Grant 2020YFB2103800 and the National Natural Science Foundation of China under Grants 61502016, 61672092.

\ifCLASSOPTIONcaptionsoff
  \newpage
\fi

\bibliography{mybibfile}

\begin{thebibliography}{10}

\bibitem{Yeung2021Image}
C.~Yeung, ``Image recognition algorithm of electrical engineering equipment
  based on machine learning method,'' in {\em 2021 IEEE International
  Conference on Power Electronics, Computer Applications (ICPECA)},
  pp.~711--714, 2021.

\bibitem{Zhang2021DDE}
R.~Zhang, A.~Albrecht, J.~Kausch, H.~J. Putzer, T.~Geipel, and P.~Halady, ``Dde
  process: A requirements engineering approach for machine learning in
  automated driving,'' in {\em 2021 IEEE 29th International Requirements
  Engineering Conference (RE)}, pp.~269--279, 2021.

\bibitem{He2021Research}
Y.~He, ``Research on the key technology of network security based on machine
  learning,'' in {\em 2021 6th International Conference on Intelligent
  Computing and Signal Processing (ICSP)}, pp.~972--975, 2021.

\bibitem{Nielsen2000Quantum}
M.~A. Nielsen and I.~L. Chuang, {\em Quantum Computation and Quantum
  Information}.
\newblock London: Cambridge University Press, 2000.

\bibitem{Arute2019Supplementary}
F.~Arute, K.~Arya, R.~Babbush, and et~al, ``Supplementary information for
  "quantum supremacy using a programmable superconducting processor",'' {\em
  Nature}, vol.~574, no.~7779, pp.~505--510, 2019.

\bibitem{Zhong2021Quantum}
H.~S. Zhong, H.~Wang, Y.~H. Deng, M.~C. Chen, L.~C. Peng, Y.~H. Luo, J.~Qin,
  D.~Wu, X.~Ding, Y.~Hu, P.~Hu, X.~Y. Yang, W.~J. Zhang, H.~Li, Y.~Li,
  X.~Jiang, L.~Gan, G.~Yang, L.~You, Z.~Wang, L.~Li, N.~L. Liu, C.~Y. Lu, , and
  J.~W. Pan, ``Quantum computational advantage using photons,'' {\em Science
  (New York, N.Y.)}, vol.~370, no.~6523, pp.~1460--1463, 2021.

\bibitem{Xin2020IEEEVPQC}
G.~Xin, J.~Han, T.~Yin, Y.~Zhou, J.~Yang, X.~Cheng, and X.~Zeng, ``Vpqc: A
  domain-specific vector processor for post-quantum cryptography based on
  risc-v architecture,'' {\em IEEE Transactions on Circuits and Systems I:
  Regular Papers}, vol.~67, no.~8, pp.~2672--2684, 2020.

\bibitem{Feynman1982Simulating}
R.~P. Feynman, ``Simulating physics with computers,'' {\em International
  Journal of Theoretical Physics}, vol.~21, no.~6-7, pp.~467--488, 1982.

\bibitem{Shor1994Algorithms}
P.~W. Shor, ``Algorithms for quantum computation: Discrete logarithms and
  factoring,'' in {\em Proceedings 35th Annual Symposium on Foundations of
  Computer Science}, pp.~124--134, 1994.

\bibitem{Grover1996fast}
L.~K. Grover, ``A fast quantum mechanical algorithm for database search,'' in
  {\em Proceedings of the 28th Annual ACM Symposium on the Theory of
  Computing}, pp.~212--219, 1996.

\bibitem{Gisin2001Quantum}
N.~Gisin, G.~Ribordy, W.~Tittel, and H.~Zbinden, ``Quantum cryptography,'' {\em
  Review of modern Physics}, vol.~74, no.~1, pp.~145--195, 2001.

\bibitem{Salim2020Enhancing}
S.~I. Salim, A.~Quaium, S.~Chellappan, and A.~B. M.~A. Al~Islam, ``Enhancing
  fidelity of quantum cryptography using maximally entangled qubits,'' in {\em
  GLOBECOM 2020 - 2020 IEEE Global Communications Conference}, pp.~1--6, 2020.

\bibitem{Kulik2022Experimental}
S.~P. Kulik, K.~S. Kravtsov, and S.~N. Molotkov, ``Experimental resources
  needed to implement photon number splitting attack in quantum cryptography,''
  {\em Laser Physics Letters}, vol.~19, no.~2, p.~025203, 2022.

\bibitem{Zhou2015Quantum}
R.~Zhou, Y.~Sun, and P.~Fan, ``Quantum image gray-code and bit-plane
  scrambling,'' {\em Quantum Information Processing}, vol.~14, no.~5,
  pp.~1--18, 2015.

\bibitem{Jiang2016Quantum}
N.~Jiang, Y.~Dang, and J.~Wang, ``Quantum image matching,'' {\em Quantum
  Information Processing}, vol.~15, no.~9, pp.~3543--3572, 2016.

\bibitem{Zhang2022Boundary}
R.~Zhang, J.~Wang, and N.~Jiang, ``Boundary extension methods to quantum signal
  mean filtering,'' {\em Quantum Information Processing}, vol.~21, no.~2, 2022.

\bibitem{Li2018Quantum}
H.~S. Li, P.~Fan, H.~Y. Xia, H.~Peng, and S.~Song, ``Quantum implementation
  circuits of quantum signal representation and type conversion,'' {\em IEEE
  Transactions on Circuits and Systems I: Regular Papers}, vol.~99, pp.~1--14,
  2018.

\bibitem{Bennett2008Quantum}
C.~Bennett and P.~Shor, ``Quantum information theory,'' {\em IEEE Transactions
  on Information Theory}, vol.~44, no.~6, pp.~2724--2742, 2008.

\bibitem{Schuld2016Prediction}
M.~Schuld, I.~Sinayskiy, and F.~Petruccione, ``Prediction by linear regression
  on a quantum computer,'' {\em Physical Review A}, vol.~94, no.~2, p.~022342,
  2016.

\bibitem{Wang2017Quantum}
G.~Wang, ``Quantum algorithm for linear regression,'' {\em Physical Review A},
  vol.~96, no.~1, p.~012335, 2017.

\bibitem{Yu2021An}
C.~H. Yu, F.~Gao, and Q.~Y. Wen, ``An improved quantum algorithm for ridge
  regression,'' {\em IEEE Transactions on Knowledge and Data Engineering},
  vol.~33, no.~3, pp.~858--866, 2021.

\bibitem{Esma2013Quantum}
E.~A\"{\i}emeur, G.~Brassard, and S.~Gambs, ``Quantum speed-up for unsupervised
  learning,'' {\em Machine Learning}, vol.~90, pp.~261--287, 2013.

\bibitem{Romero2017Quantum}
J.~Romero, J.~P. Olson, and A.~Aspuru-Guzik, ``Quantum autoencoders for
  efficient compression of quantum data,'' {\em Quantum Science and
  Technology}, vol.~2, no.~4, p.~045001, 2017.

\bibitem{Yu2019Quantum}
C.~H. Yu, F.~Gao, S.~Lin, and J.~Wang, ``Quantum data compression by principal
  component analysis,'' {\em Quantum Information Processing}, vol.~18, no.~8,
  p.~249, 2019.

\bibitem{Cong2016Quantum}
I.~Cong and L.~Duan, ``Quantum discriminant analysis for dimensionality
  reduction and classification,'' {\em New Journal of Physics}, vol.~18, no.~7,
  p.~073011, 2016.

\bibitem{Schuld2017Implementing}
M.~Schuld, M.~Fingerhuth, and F.~Petruccione, ``Implementing a distance-based
  classifier with a quantum interference circuit,'' {\em Europhysics Letters},
  vol.~119, no.~6, p.~60002, 2017.

\bibitem{2018Quantum}
M.~Schuld and F.~Petruccione, ``Quantum ensembles of quantum classifiers,''
  {\em Scientific Reports}, vol.~8, no.~1, p.~2772, 2018.

\bibitem{Dang2018Image}
Y.~Dang, N.~Jiang, H.~Hu, Z.~Ji, and W.~Zhang, ``Image classification based on
  quantum k-nearest-neighbor algorithm,'' {\em Quantum Information Processing},
  vol.~17, no.~9, p.~239, 2018.

\bibitem{Li2020Quantumneural}
P.~C. Li and B.~Wang, ``Quantum neural networks model based on swap test and
  phase estimation,'' {\em Neural Networks}, vol.~130, no.~16, pp.~152--164,
  2020.

\bibitem{Zhao2019Building}
J.~Zhao, Y.~H. Zhang, C.~P. Shao, Y.~C. Wu, G.~C. Guo, and G.~P. Guo,
  ``Building quantum neural networks based on swap test,'' {\em Physical Review
  A}, vol.~100, no.~012334, 2019.

\bibitem{Joshi2021Entanglement}
S.~K. Joshi, Z.~Huang, A.~Fletcher, N.~Solomons, I.~V. Puthoor, Y.~Pelet,
  D.~Aktas, C.~Lupo, A.~O. Quintavalle, S.~Wengerowsky, R.~S. Tessinari,
  O.~Alia, R.~Wang, M.~Clark, N.~Venkatachalam, E.~Hugues-Salas, G.~T.
  Kanellos, M.~Lončarić, S.~P. Neumann, B.~Liu, T.~Scheidl, .~Samec,
  L.~Kling, A.~Qiu, R.~Nejabati, D.~Simeonidou, E.~Andersson, S.~Pirandola,
  R.~Ursin, M.~Stipčević, and J.~G. Rarity, ``Entanglement based quantum
  networks: Protocols, ai control plane amp; coexistence with classical
  communication.,'' in {\em 2021 Conference on Lasers and Electro-Optics Europe
  European Quantum Electronics Conference (CLEO/Europe-EQEC)}, 2021.

\bibitem{Yu2008Quantum}
S.~Yu and N.~Ma, ``Quantum neural network and its application in vehicle
  classification,'' in {\em 2008 Fourth International Conference on Natural
  Computation}, vol.~2, pp.~499--503, 2008.

\bibitem{Cavallaro2020Approaching}
G.~Cavallaro, D.~Willsch, M.~Willsch, K.~Michielsen, and M.~Riedel,
  ``Approaching remote sensing image classification with ensembles of support
  vector machines on the d-wave quantum annealer,'' in {\em IGARSS 2020 - 2020
  IEEE International Geoscience and Remote Sensing Symposium}, pp.~1973--1976,
  2020.

\bibitem{Allcock2020A}
J.~Allcock and C.~Y. Hsieh, ``A quantum extension of svm-perf for training
  nonlinear svms in almost linear time,'' {\em Quantum}, vol.~4, no.~342, 2020.

\bibitem{Kerenidis2021Quantum}
I.~Kerenidis, A.~Prakash, and S.~D\'{a}niel, ``Quantum algorithms for
  second-order cone programming and support vector machines,'' {\em Quantum},
  vol.~5, no.~427, 2021.

\bibitem{Li2015Experimental}
Z.~K. Li, X.~M. Liu, N.~Y. Xu, and J.~F. Du, ``Experimental realization of a
  quantum support vector machine,'' {\em Physical Review Letters}, vol.~114,
  no.~14, p.~140504, 2015.

\bibitem{Havenstein2018Comparisons}
C.~Havenstein, D.~Thomas, and S.~Chandrasekaran, ``Comparisons of performance
  between quantum and classical machine learning,'' {\em Smu Data Science
  Review}, vol.~1, no.~4, 2018.

\bibitem{Ye2020Quantum}
Z.~Ye, L.~Li, H.~Situ, and Y.~Wang, ``Quantum speedup of twin support vector
  machines,'' {\em Science China: Information Sciences}, vol.~63, p.~189501,
  2020.

\bibitem{Willsch2020Support}
D.~Willsch, M.~Willsch, H.~De~Raedt, and M.~K., ``Support vector machines on
  the d-wave quantum annealer,'' {\em Computer Physics Communications},
  vol.~248, no.~107006, 2020.

\bibitem{Anguita2003Quantum}
D.~Anguita, S.~Ridella, F.~Rivieccio, and R.~Zunino, ``Quantum optimization for
  training support vector machines,'' {\em Neural Networks}, vol.~16, no.~5-6,
  pp.~763--770, 2003.

\bibitem{Rebentrost2014Quantum}
P.~Rebentrost, M.~Mohseni, and S.~Lloyd, ``Quantum support vector machine for
  big data classification,'' {\em Physical Review Letters}, vol.~113, no.~13,
  p.~130503, 2014.

\bibitem{Buhrman2001Quantum}
H.~Buhrman, R.~Cleve, J.~Watrous, and R.~Wolf, de, ``Quantum fingerprinting,''
  {\em Physical Review Letters}, vol.~87, no.~16, p.~167902, 2001.

\bibitem{Garcia-Escartin2013swap}
J.~C. Garcia-Escartin and P.~Chamorro-Posada, ``Swap test and hong-ou-mandel
  effect are equivalent,'' {\em Physical Review A}, vol.~87, no.~5,
  pp.~239--257, 2013.

\bibitem{Bishwas2018An}
A.~K. Bishwas, A.~Mani, and V.~Palade, ``An all-pair quantum svm approach for
  big data multiclass classification,'' {\em Quantum Information Processing},
  vol.~17, p.~282, 2018.

\bibitem{Bishwas2016Big}
A.~K. Bishwas, A.~Mani, and V.~Palade, ``Big data classification with quantum
  multiclass svm and quantum one-against-all approach,'' in {\em 2016 2nd
  International Conference on Contemporary Computing and Informatics (IC3I)},
  pp.~875--880, 2016.

\bibitem{Windridge2018Quantum}
D.~Windridge, R.~Mengoni, and R.~Nagarajan, ``Quantum error-correcting output
  codes,'' {\em International Journal of Quantum Information}, vol.~16, no.~08,
  p.~1840003, 2018.

\bibitem{Feng2019Quantum}
X.~Feng, J.~C. Li, C.~G. Huang, J.~Z. Li, R.~Y. Chen, J.~F. Ke, and Z.~J. Ma,
  ``Quantum algorithm for support vector machine with exponentially improved
  dependence on precision,'' in {\em International Conference on Artificial
  Intelligence and Security}, pp.~578--587, 2019.

\bibitem{Suykens1999LeastSquares}
J.~Suykens and J.~Vandewalle, ``Least squares support vector machine
  classifiers,'' {\em Neural Processing Letters}, vol.~9, no.~3, pp.~293--300,
  1999.

\bibitem{Aram2009Quantum}
A.~W. Harrow, A.~Hassidim, and S.~Lloyd, ``Quantum algorithm for solving linear
  systems of equations,'' {\em Physical Review Letters}, vol.~15, no.~103,
  p.~150502, 2009.

\bibitem{Hou2020Quantum}
Y.~Y. Hou, J.~Li, X.~B. Chen, H.~J. Li, , C.~Y. Li, Y.~Tian, L.~L. Li, Z.~W.
  Cao, and N.~Wang, ``Quantum algorithm for help-training semi-supervised
  support vector machine,'' {\em Quantum Information Processing}, vol.~19,
  no.~9, 2020.

\bibitem{Bishwas2020Gaussian}
A.~K. Bishwas, A.~Mani, and V.~Palade, ``Gaussian kernel in quantum learning,''
  {\em International Journal of Quantum Information}, vol.~18, no.~03,
  pp.~1019--1041, 2020.

\bibitem{Schuld2018Quantum}
M.~Schuld and N.~Killoran, ``Quantum machine learning in feature hilbert
  spaces,'' {\em Physical Review Letters}, vol.~122, no.~4, p.~040504, 2019.

\bibitem{Brassard2000Quantum}
G.~Brassard, P.~H${\o}$yer, M.~Mosca, and A.~Tapp, {\em Quantum Amplitude
  Amplification and Estimation}, vol.~9305.
\newblock Contemporary mathematics Series Millenium, 2002.

\bibitem{Vazquez2021Efficient}
A.~C. Vazquez and S.~Woerner, ``Efficient state preparation for quantum
  amplitude estimation,'' {\em Physical Review Applied}, vol.~15, p.~034027,
  2021.

\bibitem{Grinko2021Iterative}
D.~Grinko, J.~Gacon, C.~Zoufal, and S.~Woerner, ``Iterative quantum amplitude
  estimation,'' {\em npj Quantum Information}, vol.~7, no.~52, 2021.

\bibitem{Montanaro2015Quantum}
A.~Montanaro, ``Quantum speedup of monte carlo methods,'' {\em Proceedings of
  the Royal Society of London A: Mathematical, Physical and Engineering
  Sciences}, vol.~471, no.~2181, 2015.

\bibitem{Rebentrost2018Quantum}
P.~Rebentrost, B.~Gupt, and T.~R. Bromley, ``Quantum computational finance:
  Monte carlo pricing of financial derivatives,'' {\em Physical Review A},
  vol.~98, no.~022321, 2018.

\bibitem{Miyamoto2020Reduction}
K.~Miyamoto and K.~Shiohara, ``Reduction of qubits in quantum algorithm for
  monte carlo simulation by pseudo-random number generator,'' {\em Physics
  Review A}, vol.~102, no.~022424, 2020.

\bibitem{Miyahara2019Quantum}
H.~Miyahara, K.~Aihara, and W.~Lechner, ``Quantum expectation-maximization
  algorithm,'' {\em Physics Review A}, vol.~101, no.~012326, 2020.

\bibitem{Kerenidis2019q}
I.~Kerenidis, J.~Landman, A.~Luongo, and A.~Prakash, ``q-means: a quantum
  algorithm for unsupervised machine learning,'' {\em Advances in Neural
  Information Processing Systems}, no.~372, pp.~4134--4144, 2019.

\bibitem{Kerenidis2020Quantum}
I.~Kerenidis and A.~Prakash, ``Quantum gradient descent for linear systems and
  least squares,'' {\em Physics Review A}, vol.~102, p.~022316, 2020.

\bibitem{Lin2020Quantum}
J.~Lin, D.~B. Zhang, S.~Zhang, T.~Li, X.~Wang, and W.~S. Bao,
  ``Quantum-enhanced least-square support vector machine: Simplified quantum
  algorithm and sparse solutions,'' {\em Physics Letters A}, vol.~384, no.~25,
  p.~126590, 2020.

\bibitem{Li2019Quantum}
P.~C. Li, J.~H. Guo, B.~Wang, and M.~Q. Hao, ``Quantum circuits for calculating
  the squared sum of the inner product of quantum states and its application,''
  {\em International Journal of Quantum Information}, vol.~17, no.~3,
  p.~1950043, 2019.

\bibitem{Cortes1995Suppor}
C.~Cortes and V.~Vapnik, ``Suppor-vector networks,'' {\em Machine Learning},
  vol.~20, no.~3, pp.~273--297, 1995.

\bibitem{Smola2000Generalized}
A.~J. Smola, P.~Bartlett, B.~Sch\"{o}lkopf, and D.~Schuurmans, ``Generalized
  support vector machines,'' {\em Advances in Large Margin Classifiers},
  pp.~135--146, 2000.

\end{thebibliography}
\bibliographystyle{ieeetr}


%
%
%

%


\begin{IEEEbiography}[{\includegraphics[width=1in,height=1.25in,clip,keepaspectratio]{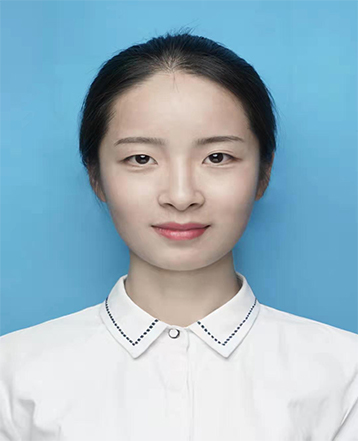}}]{Rui Zhang}
was born in Minquan County, Henan Province, China in 1994. She received the M.S. degree in School of Mathematics and Statistics, Henan University, China, in 2020. She is currently pursuing the Ph.D. degree in the School of Computer and Information Technology, Beijing Jiaotong University, China. Her research interest includes quantum computing and machine learning.
\end{IEEEbiography}

\begin{IEEEbiography}[{\includegraphics[width=1in,height=1.25in,clip,keepaspectratio]{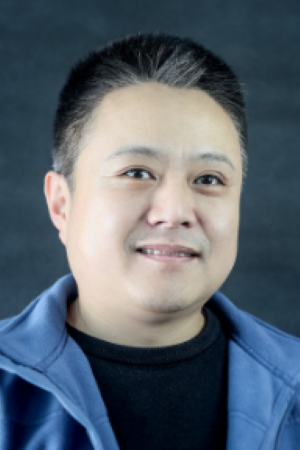}}]{Jian Wang}  received the Ph.D. degree in computer science from Beijing University of Posts and Telecommunications, Beijing, China, in 2008. Since 2008, he has been an associate professor with School of Computer and Information Technology, Beijing, China. His research interest includes quantum computing and quantum machine learning, data privacy.
\end{IEEEbiography}


\begin{IEEEbiography}[{\includegraphics[width=1in,height=1.25in,clip,keepaspectratio]{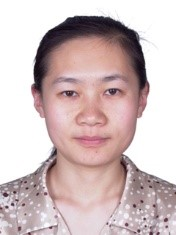}}]{Nan Jiang} received the Ph.D. degree in computer science from the Beijing University of Posts and Telecommunications, China, in 2006. From 2015 to 2016, she was a Visiting Scholar with the College of Science, Purdue University, USA. She is currently an Professor of computer science with the Beijing University of Technology, and the Beijing Key Laboratory of Trusted Computing. Her research interests include quantum image processing, quantum machine learning and information hiding.
\end{IEEEbiography}

\begin{IEEEbiography}[{\includegraphics[width=1in,height=1.25in,clip,keepaspectratio]{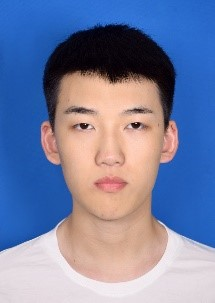}}]{Zichen Wang} was born in Xuzhou of Jiangsu Province. He graduated from Xuzhou University of Technology. Now studying in Beijing University of technology for a master's degree. The current research direction is quantum machine learning.
\end{IEEEbiography}




\end{document}